\documentclass[a4paper]{article} %

\usepackage[english]{babel}
\usepackage{
        amsmath,
        amsfonts,
        amssymb,
        amsthm,
        url,
        graphicx,
        mathtools,
        upgreek,
        todonotes
}

\usepackage{url}
\usepackage{mathpazo}
\usepackage[T1]{fontenc}

\usepackage[left=3cm,top=3cm,right=3cm,bottom=3cm]{geometry}

\usepackage[title]{appendix}

\usepackage{xcolor}

\usepackage{hyperref}
\hypersetup{
colorlinks,
linkcolor={red!50!black},
citecolor={blue!50!black},
urlcolor={blue!80!black}
}

\graphicspath{{fig/}}

\newtheorem{proposition}{\bf Proposition}[section]

\newtheorem{remark}{\bf Remark}[section]

\DeclareMathOperator{\diver}{div}
\DeclareMathOperator{\grad}{grad}
\DeclareMathOperator{\Diver}{Div}
\DeclareMathOperator{\Grad}{Grad}

\DeclareMathOperator{\sym}{Sym}

\DeclareMathOperator{\lin}{Lin}
\DeclareMathOperator{\orth}{Orth}
\DeclareMathOperator{\tr}{tr}

\DeclareMathOperator{\diag}{diag}

\usepackage{mathtools}
\usepackage{pgfplots}
\pgfplotsset{/pgf/number format/use comma,compat=newest}

\renewcommand\epsilon{\varepsilon}

\newcommand{\vect}[1]{\boldsymbol{\mathbf{#1}}}
\newcommand{\tens}[1]{\mathsf{#1}}

\newcommand{\sistema}[1]{\left\{\begin{aligned}#1\end{aligned}\right.}
\renewcommand{\d}{\mathrm{d}}

\newcommand{\F}{\tens{F}}
\newcommand{\B}{\tens{B}}
\newcommand{\C}{\tens{C}}
\newcommand{\Fe}{\tens{F}_\text{e}}
\newcommand{\Be}{\tens{B}_\text{e}}
\newcommand{\Ce}{\tens{C}_\text{e}}
\newcommand{\oF}{\overline{\F}}

\newcommand{\Sigmad}{\tens{\Upsigma}_\text{d}}
\newcommand{\Sigmav}{\tens{\Upsigma}_\text{v}}

\begin{document}
\title{\textsc{Modelling of initially stressed solids: structure of the energy density in the incompressible limit}}

\author{\textsc{M. Magri}\thanks{\href{mailto:marco.magri@polimi.it}{\texttt{marco.magri@polimi.it}}} $\,\,\cdot$
\textsc{D. Riccobelli}\thanks{\href{mailto:davide.riccobelli@polimi.it}{\texttt{davide.riccobelli@polimi.it}}}\bigskip\\
\normalsize MOX -- Dipartimento di Matematica, Politecnico di Milano,\\
\normalsize piazza Leonardo da Vinci 32, Milano, Italy}
\date{\today}

\maketitle

\begin{abstract}
This study addresses the modelling of elastic bodies, particularly when the relaxed configuration is unknown or non-existent. We adopt the theory of initially stressed materials, incorporating the deformation gradient and stress state of the reference configuration (initial stress tensor) into the response function. We show that for the theory to be applicable, the response function of the relaxed material is invertible up to an element of the material symmetry group. Additionally, we establish that commonly imposed constitutive restrictions, namely the initial stress compatibility condition and initial stress reference independence, naturally arise when assuming the initial stress is generated solely from elastic distortion. The paper delves into modelling aspects concerning incompressible materials, showcasing the expressibility of strain energy density as a function of the deviatoric part of the initial stress tensor and the isochoric part of the deformation gradient. This not only reduces the number of independent invariants in the energy functional, but also enhances numerical robustness in finite element simulations. The findings of this research hold significant implications for modelling materials with initial stress, extending potential applications to areas such as mechanobiology, soft robotics, and 4D printing.
\end{abstract}

\section{Introduction}
In the modelling of solid materials, the elastic strain is typically defined with respect to a chosen configuration known as the reference configuration. While it is commonly selected as a configuration free of mechanical stress, certain situations necessitate adopting a reference configuration that exhibits a non-zero stress state. The stress distribution in this reference configuration is called  \emph{initial stress}.

For example, biological tissues are frequently subject to external mechanical force (e.g. the blood pressure on vessel walls and on the heart chambers) and their stress-free configuration is not easily observable \cite{barnafi2024reconstructing}.
An initial stress generated by the imposition of external loads is also known as \textit{pre-stress} \cite{10.1063/1.4704566}. Pre-stress magnitude and distribution can be designed to enhance the mechanical response of materials, as seen in pre-stressed concrete structures \cite{concrete}.

Moreover, materials may store mechanical stress even in the absence of external forces. In such a case the stress state is referred to as \textit{residual stress} \cite{Hoger1985}. In general, residual stresses are generated by geometrical incompatibilities at the microstructural level that necessitate the introduction of elastic distortions to maintain the continuity of the body \cite{goriely2017mathematics}.
For instance, the formation of residual stress in biological materials is driven by active processes, e.g. growth, that generate the incompatibility \cite{doi:10.1073/pnas.1213353109,Ciarletta2016}. Residual stresses may also arise in engineering materials, in this case as a result of their manufacturing process \cite{doi:10.1177/002199839302701402,ZOBEIRY201543,koss1971thermally,moghimi2017residual,damioli2022transient}. 

Regardless of its origin, a suitable design of the initial stress distribution can be exploited to enhance the resistance of the material to the deformation, e.g. in arteries \cite{chuong1986residual,holzapfel2010modelling}, or to control the shape morphing of the object \cite{ciarletta2016morphology, riccobelli2018shape}, with possible applications to soft robotics \cite{alameen2023mechanics} and 4D printing \cite{sydney2016biomimetic,zurlo2017printing}.

A successful approach to modelling initially stressed material is based on the introduction of a \textit{virtual relaxed state} of the body \cite{rodriguez1994stress}. The underlying idea is based on a multiplicative decomposition of the deformation gradient $\tens{F}=\tens{F}_\text{e}\tens{G}$, where $\tens{G}$ maps the reference configuration to a relaxed state while $\tens{F}_\text{e}$ accounts for the local elastic distortion. In general, $\tens{G}$ may not be the gradient of a deformation. This introduces incompatibilities that generate the initial stress within the body. Despite the huge success of this theory, the main drawback is that $\tens{G}$ must be constitutively provided \cite{Epstein_2012,goriely2017mathematics}. Experimental measurements of $\tens{G}$ usually rely on cutting the body to release the mechanical stress and reveal its relaxed state. Clearly, such a destructive technique is not applicable in several scenarios, such as the in-vivo estimation of the stress state of a biological tissue.

However, recent scientific developments have made it possible to measure the initial stress state of an object using non-destructive methods, such as acoustoelastic techniques \cite{li2020ultrasonic,zhang2023noninvasive}.
Based on these progresses, it is more effective to describe the material behaviour as a function of the initial stress tensor and of the deformation gradient. Such an approach is known as the \emph{theory of initially stressed materials}. The first attempts in this direction date back to the seminal works of Hoger and colleagues \cite{Hoger1986,Hoger1993,Hoger1996,Hoger1993b,Johnson1993,Johnson1995,Johnson1998}.
Further extensions have been proposed in recent years \cite{Saravanan_2008,Shams_2011,MERODIO201343,Shams_Ogden_2012,NAM201688,du2018modified,huang2021mathematical,gower2024elastic}. Specifically, Shams et al. \cite{Shams_2011} proposed a constitutive law for a hyperelastic and isotropic solid based on ten invariants involving the elastic deformation and the initial stress tensor.

More recently, several constitutive restrictions have been proposed to guide the construction of physically admissible energy functions, such as the \textit{initial stress compatibility} and the \textit{initial stress reference independence} conditions \cite{Gower_2015,Gower_2017}. Despite these advances, a unified theory for constitutive modelling of initially stressed solids remains a debated topic \cite{Ogden_2023}. The approach proposed by Gower et al.  \cite{Gower_2015} enables the derivation of a general constitutive law that holds for any reference configuration independently of the initial stress. On the contrary, other constitutive laws that do not satisfy these constraints are inherently tied to a specific reference configuration, with material properties potentially depending on the residual stress distribution \cite{Riccobelli_2019}.\\
Furthermore, the finite element simulations of incompressible initially stressed solids exhibit numerical issues, as highlighted, for example, in \cite{ciarletta2016morphology}. These problems are related to locking phenomena that frequently affect the simulation of incompressible elastic media (see chapter 15 of \cite{de2011computational}). A common approach to circumvent these drawbacks in conventional hyperelasticity is to exploit a decomposition of the stress tensors into a volumetric and a deviatoric part \cite{bonet1997nonlinear,holzapfel2000nonlinear}.

Based on the aforementioned aspects, the purpose of this paper is two-fold:
\begin{itemize}
    \item derive a comprehensive theory of initially stressed materials where the initial stress is generated solely by an elastic distortion,
    \item enhance current description of incompressible initially stressed media exploiting the volumetric-deviatoric splitting.
\end{itemize}
The paper is organised as follows: in Section~\ref{sec:init-stress}, we review the basic theory of initially stressed materials, deriving it starting by solely assuming that the initial stress is generated by an elastic distortion. Specifically, we show that commonly enforced restrictions naturally follow from this assumption.
In Section~\ref{sec:structure}, we provide a representation formula for the strain energy of an isotropic initially stressed material. In Section~\ref{sec:incompressible} we study the incompressible limit of the proposed theoretical framework. As an application, the bending of an initially stressed block is analysed in Section~\ref{sec:bending}.

\section{Initially stressed materials}
\label{sec:init-stress}
In the following, we introduce the essential notation and theory to describe the response of initially stressed materials by exploiting the theoretical framework developed by Shams et al. \cite{Shams_2011}.

We consider a reference configuration $\Omega_0\subset\mathbb{E}^3$, where $\mathbb{E}^3$ denotes the three-dimensional Euclidean space. Let $\vect{X}\in \Omega_0$ be a point of the reference configuration. The deformation field is denoted by the map $\vect{\chi}:\Omega_0\rightarrow\mathbb{E}^3$, with the current position vector given by $\vect{x} = \vect{\chi}(\vect{X})$. The deformation gradient tensor is denoted by $\tens{F} = \Grad\vect{\chi}$.

Let $\tens{T}$ and $\psi$ be the Cauchy stress tensor and the strain energy density per unit reference volume, respectively. We recall that the first Piola-Kirchhoff stress tensor $\tens{S}$ is related to $\tens{T}$ by the formula

\begin{equation}
    \label{eq:Cauchy}
    \tens{T} = J^{-1}\tens{S}\tens{F}^T,
\end{equation}
where $J=\det\tens{F}$.

When modelling initially stressed materials, it is usually assumed that both the strain energy density and the Cauchy stress tensor depend on the initial stress tensor $\tens{\Sigma}$ and the deformation gradient, that is
\begin{equation}
    \label{eq:psiFSigma}
    \psi = \widehat{\psi}(\tens{F},\, \tens{\Sigma}) \quad \text{and} \quad \tens{T} = \widehat{\tens{T}}(\tens{F},\,\tens{\Sigma}).
\end{equation}
Before proceeding into this direction, we present a simple one-dimensional example to illustrate the procedure.

\subsection{An illustrative example: an initially stressed elastic spring}
\label{sec:0d}
We consider an elastic spring in its relaxed state. In this configuration, its length $L$ coincides with its rest length. We can write its constitutive response out of this stress free state as simply as
\begin{equation}
    \label{eq:spring}
    F=f(\lambda) \, ,
\end{equation}
where $F$ is the elastic force generated by the spring, $\lambda=\ell/L$ is the stretch, and $\ell$ is the current length of the spring (see Fig.~\ref{fig:molle}). We further assume $f$ to be invertible.

We now consider a different configuration of the spring with length $L_1 \neq L$. According to \eqref{eq:spring}, an elastic force $\tau_1=f(e_1)$, with $e_1=L_1/L$, is associated with this configuration. Our aim is to describe the response of the spring out of this stressed configuration. Accordingly, we shall introduce a constitutive function $f_1$ that depends exclusively on the stretch $\lambda_1=\ell/L_1$ and on the force $\tau_1$. Since $\lambda = e_1\lambda_1=f^{-1}(\tau_1)\lambda_1$, we have

\begin{equation}
    \label{eq:f1}
    F = f_1(\lambda_1,\,\tau_1) =f(f^{-1}(\tau_1)\lambda_1).
\end{equation}
Similarly, we can take another configuration with length $L_2 \neq L_1$. By following the same procedure used above, we can build a constitutive function $f_2$ out of this configuration as follows
\begin{equation}
    \label{eq:f2}
    F=f_2(\lambda_2,\,\tau_2) = f(f^{-1}(\tau_2)\lambda_2).
\end{equation}
where $\tau_2=f(e_2)$, with $e_2=L_2/L$, and the stretch $\lambda_2=\ell/L_2$.
\begin{figure}[t!]
    \centering
    \includegraphics[width=0.8\textwidth]{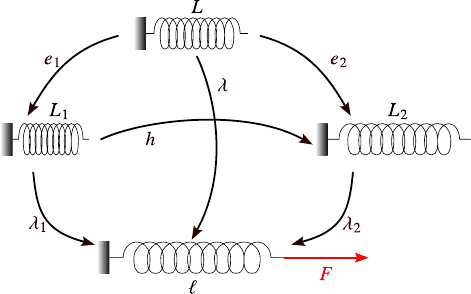}
    \caption{A scheme of the example of section~\ref{sec:0d}.}
    \label{fig:molle}
\end{figure}

At this point, we observe that $f_1$ and $f_2$ defined in \eqref{eq:f1} and \eqref{eq:f2} are merely the same response function that is evaluated for different values. To support this, we introduce $h=L_2/L_1=\lambda_1\lambda_2^{-1}$, so it can be easily shown that
\[
    F = f_1(\lambda_2 h,\tau_1) = f_2 (\lambda_2,\,\tau_2) =f_1 (\lambda_2,\,\tau_2) \, .
\]

Moreover, since $\tau_2 = f_1(h,\,\tau_1)$, we get
\begin{equation}
    \label{eq:ISRI1D}
    f_1(\lambda_2 h,\tau_1)=f_1 (\lambda_2,\,f_1(h,\,\tau_1)) \, .
\end{equation}
This last equation is the one-dimensional equivalent of the \textit{initial stress reference independence} condition proposed in \cite{Gower_2017} for initially stressed materials. A key element in this derivation is the assumption of a common relaxed state of the initially stressed spring. Although such an example is an over-simplification of a generic multidimensional problem, it enlightens two key assumptions:
\begin{enumerate}
    \item the invertibility of the response function $f$;
    \item the fact that the initial stress state is generated by an elastic distortion.
\end{enumerate}
We now extend this considerations to a three-dimensional elastic body with initial stress to study the general structure of the problem.

\subsection{Initially stressed three-dimensional elastic bodies: fundamental assumptions} \label{sec:assumptions}

To derive a general theory for initially stressed three-dimensional bodies, we consider a natural material with a purely hyperelastic response. The following notation will be used thereafter

\begin{itemize}
    \item $\lin$ denotes the set of linear second order tensors.
    \item $\lin^+$ denotes the set of all the $\tens{F}\in \lin$ such that $\det\tens{F}>0$.
    \item $\sym$ is the set of all the symmetric second order tensors.
    \item $\orth$ and $\orth^+$ are the sets of the orthogonal matrices and of rotations, respectively.
\end{itemize}

\subsubsection{Invertibility of the response function}
\label{sec:ass_1}

As pursued in Sec. \ref{sec:0d},  we start by considering a constitutive response of a hyperelastic material free of initial stresses i.e. we take a response function $\tens{T}_0(\Fe)$, where $\Fe\in\lin^+$ is the elastic distortion measured with respect to a stress-free state, while  $\tens{T}_0$ is the Cauchy stress tensor. Moreover, we denote by $\psi_0(\Fe)$ the strain energy density associated with $\tens{T}_0(\Fe)$.

For a given stress state $\tens{\Sigma}\in\sym$, we would like to identify a distortion $\tens{P}_\tens{\Sigma}$ such that $\tens{T}_0(\tens{P}_\tens{\Sigma})=\tens{\Sigma}$. The existence of such a distortion $\tens{P}_\tens{\Sigma}$ has been addressed in \cite{Riccobelli_2019} provided that the strain energy density is non degenerate, that is

\[
    \left\{
    \begin{aligned}
         & \psi_0(\Fe)\rightarrow+\infty\qquad\text{as }|\Fe|+|\Fe^{-1}|+\det\Fe\rightarrow+\infty, \\
         & \psi_0(\Fe)\rightarrow+\infty\qquad\text{as }\det\Fe\rightarrow 0.
    \end{aligned}
    \right.
\]
Such an assumption requires the strain energy density to blow up when the body is locally subject to an extreme deformation (e.g. one of the principal stretches goes to infinity). However, the simple existence of $\tens{P}_\tens{\Sigma}$ is not enough, as we show next.

Compared with the illustrative example proposed in Sec.~\ref{sec:0d}, we cannot have an exact inverse of $\tens{T}_0$ in this 3D case. Indeed, for any $\tens{Q}\in\mathcal{G}$, where $\mathcal{G}\subset\orth^+$ is the material symmetry group, we have
\[
    \tens{T}_0(\Fe\tens{Q})=\tens{T}_0(\Fe)   \, ,
\]
so that $\tens{T}_0$ is not injective, in general. Nevertheless, it is possible to introduce a relation of equivalence $\sim$ in $\lin^+$, such that
\[
    \Fe^{(1)}\sim\Fe^{(2)} \quad\Leftrightarrow\quad \exists\tens{Q}\in\mathcal{G}\;|\;\Fe^{(1)} = \Fe^{(2)}\tens{Q} \, .
\]

We further define the quotient set $\mathcal{F}=\lin^+/\sim$ composed of the equivalence classes $[\Fe]=\{\tens{A}\in\lin^+\;|\;\tens{A}\sim\Fe\}$ induced by $\sim$, namely
\[
    \mathcal{F} = \{[\Fe],\,\Fe\in\lin^+\} \, .
\]
We can now conveniently construct an alternative response function $\tens{T}_\mathcal{F}$, as follows
\[
    \begin{aligned}
         & \tens{T}_\mathcal{F}:\mathcal{F}\rightarrow\sym, \\
         & \tens{T}_\mathcal{F}([\Fe]):=\tens{T}_0(\Fe),
    \end{aligned}
\]
which is, of course, well-defined since $\tens{T}_0(\tens{A})$ is independent of the choice of $\tens{A}\in [\Fe]$. Differently from $\tens{T}_0$, the newly introduced response function $\tens{T}_\mathcal{F}$ can be bijective under particular circumstances.

\begin{remark}
    \label{rem:invert}
    Requiring the invertibility of the response function over the set $\mathcal{F}$ is restrictive.  Consider, for instance, an incompressible material whose strain energy density is given by
    \[
        \psi_0(\Fe) = \mu(\sqrt{\lambda_1}+\sqrt{\lambda_2}+\sqrt{\lambda_3}-3),
    \]
    Such an energy is representative of material softening and, therefore, does not lead to an injective expression of the principal stresses. We show this fact  by taking a uniaxial deformation in plane strain conditions, such that $\Fe = \diag(\lambda,\,\lambda^{-1},\,1)$. The principal stress along the axial direction can be computed as $t_1(\lambda)=\partial\psi_0(\Fe)/\partial\lambda$, resulting in
    \[
        t_1(\lambda) = \mu\frac{(\lambda -1) }{2 \lambda ^{3/2}},
    \]
    which is clearly non-bijective. \textbf{This behaviour is  linked with the existence of multiple deformations associated with the same stress response, but having different strain energy}. For instance, we have
    \[
        t_1\left(2\right) = t_1\left(3 + \sqrt{5}\right),
    \]
    but
    \[
        \left.\psi_0(\Fe)\right|_{\lambda=2}=\frac{3 \sqrt{2}}{2}\mu<3 \frac{\sqrt{10}+\sqrt{2}}{2} \mu = \left.\psi_0(\Fe)\right|_{\lambda=3 + \sqrt{5}}.
    \]
\end{remark}

In order to construct a robust theory, in the following \emph{we require $\tens{T}_\mathcal{F}$ to be bijective.}

This discussion simplifies a lot if we assume an isotropic material response, that is $\mathcal{G}=\orth^+$.
Accordingly, using the polar decomposition theorem $\Fe=\tens{V}\tens{R}$, with $\tens{V}$ being a symmetric positive definite tensor and $\tens{R}\in\orth^+$, we have
\[
    \tens{T}_0 (\Fe) = \tens{T}_0(\tens{V}).
\]
Here, $\tens{V}$ can be seen as the representative of its equivalence class $[\Fe]$. For isotropic materials, the principal directions of $\tens{T}_0$ and of $\tens{V}$ are the same, therefore
\[
    t_j=f_j(\lambda_1,\,\lambda_2,\,\lambda_3),\quad j\in\{1,\,2,\,3\},
\]
where $t_j$ and $\lambda_j$ are the principal values of $\tens{T}$ and $\tens{V}$, respectively.
Consequently, the requirement that $\tens{T}_\mathcal{F}$ be invertible, here, reduces to the invertibility of the three scalar functions $t_j$.

\subsubsection{Initially stressed configurations}
\label{sec:ass_2}

We are now interested to model the behaviour of an elastic body using a reference configuration which is not stress-free. In particular, the constitutive behaviour of each point will be characterized by two properties:
\begin{itemize}
    \item the initial stress state in that point, indicated by $\tens{\Sigma}$.
    \item the elastic properties of the material. These are encoded in the response function $\tens{T}_0(\Fe)$, which represents the behaviour of the material in the absence of initial stress.
\end{itemize}

In particular, we assume that the initial stress $\tens{\Sigma}$ is generated by an elastic distortion. While this aspect may seem to be restrictive, this hypothesis allows us to model a large variety of processes, such as the pre-stretching of elastic bodies and the residual stress induced by geometrical frustrations. In both cases, such a process can be conceptualised as the application of a local distortion $\tens{P}(\tens{\Sigma})$ to a small neighbourhood of a stress-free material before being placed in the reference configuration $\Omega_0$.

From a mathematical point of view, this requirement can be enforced by requiring that the strain energies $\psi_0(\Fe)$ and $\psi(\tens{F},\,\tens{\Sigma})$ are materially isomorphic \cite{Noll_1967,Epstein_2012,Epstein_2015} or, by using the nomenclature introduced by Epstein, we take $\psi_0(\Fe)$ as the \emph{archetype} of the material \cite{Epstein_2001,Epstein_2002}. Operationally, given an initial stress $\tens{\Sigma}$ at a material point $\vect{X}$, we define its related distortion $\tens{P}$ as
\[
    \tens{P}:\sym\rightarrow\lin^+ \, ,
\]
with
\[
    \boxed{
        \begin{gathered}
            \tens{P}=\mathcal{C}_\mathcal{F}\circ\tens{T}_\mathcal{F}^{-1}\\
            \tens{P}_\tens{\Sigma}=\tens{P}(\tens{\Sigma})
        \end{gathered}
    }
\]
where $\mathcal{C}_\mathcal{F}$ is the choice function that associates to any set $[\tens{P}]\in \mathcal{F}$ one of its elements\footnote{The existence of such a function is guaranteed by the axiom of choice.}. Clearly, by definition
\begin{equation}
    \tens{T}_0(\tens{P}(\tens{\Sigma}))=\tens{\Sigma}.
\end{equation}
Consequently, by the definition of material isomorphism, we get
\begin{equation}
    \label{eq:ISRIpsi}
    \widehat{\psi}(\tens{F},\,\tens{\Sigma}) = \frac{1}{\det\tens{P}(\tens{\Sigma})}\,\psi_0(\tens{F}\tens{P}(\tens{\Sigma})).
\end{equation}
Moreover, by denoting $\tens{F}_\text{e} = \tens{F}\tens{P}_{\tens{\Sigma}}$, we get
\begin{equation}
    \label{eq:TFSigma}
    \widehat{\tens{T}}(\tens{F},\,\tens{\Sigma})=\frac{1}{\det \tens{F}\det\tens{P}(\tens{\Sigma})}\left.\frac{\partial\psi_0}{\partial\tens{F}_\text{e}}\right|_{\tens{F}_\text{e}=\tens{F}\tens{P}(\tens{\Sigma})}\tens{P}(\tens{\Sigma})^{T}\tens{F}^T=\tens{T}_0(\tens{F}\tens{P}(\tens{\Sigma})).
\end{equation}

\subsection{Initially stressed materials: properties of the constitutive law}

From equations \eqref{eq:ISRIpsi} and \eqref{eq:TFSigma}, we can immediately retrieve two constitutive restrictions that are usually enforced as fundamental assumptions in modelling initially stressed materials.

\subsubsection{The initial stress compatibility condition}

By taking $\tens{F}=\tens{I}$ in \eqref{eq:TFSigma} we trivially get
\begin{equation}
    \label{eq:ISCC}
    \widehat{\tens{T}}(\tens{I},\,\tens{\Sigma}) = \tens{T}_0(\tens{P}(\tens{\Sigma})) = \tens{\Sigma}.
\end{equation}

Therefore, this equation must hold for all $\tens{\Sigma}$ belonging to $\sym$. Such a property, is referred to as the \textit{initial stress compatibility condition} \cite{Shams_2011,Gower_2015}.

\subsubsection{The initial stress reference independence}
A much discussed and open issue is related to the properties of \eqref{eq:psiFSigma} while making a change of reference configuration. According to \cite{Gower_2015,Gower_2017}, this operation provides a constitutive restriction called \emph{initial stress reference independence} (ISRI). However, such a restriction has received some criticism, e.g. \cite{Ogden_2023}, mainly because it is seen as a condition that unnecessarily restricts the admissible constitutive laws of initially stressed materials. To clarify this aspect, we show here that ISRI is a mere consequence of the assumptions we made in Section~\ref{sec:assumptions}.

\begin{figure}
    \centering
    \includegraphics[width=0.8\textwidth]{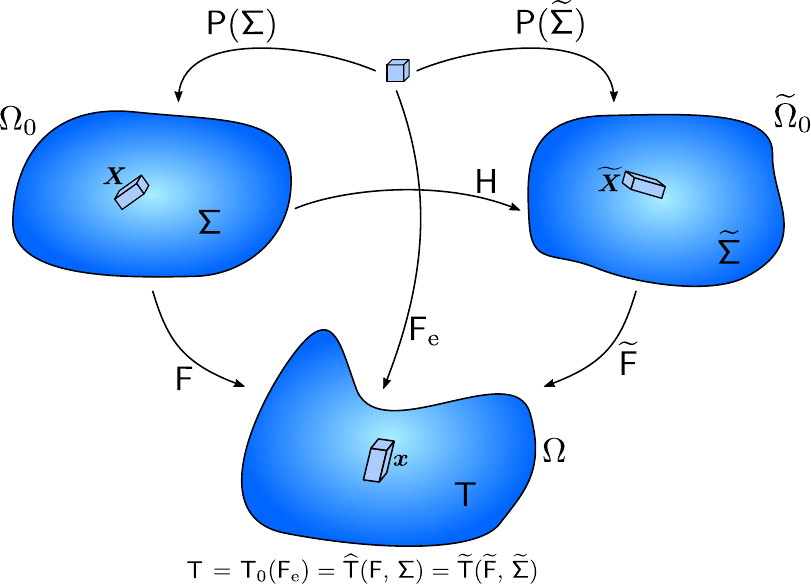}
    \caption{Change of reference configuration for initially stressed materials.}
    \label{fig:ref_conf}
\end{figure}

Let $\widetilde{\Omega}_0$ be a different reference configuration than $\Omega_0$, with $\vect{\xi}:\Omega_0\rightarrow\widetilde{\Omega}_0$ indicating the map between these two reference configurations. We also introduce $\tens{H}=\Grad\vect{\xi}$ and denote with $\widetilde{\tens{F}}$ the gradient of the deformation field $\widetilde{\vect{\chi}}:\widetilde{\Omega}_0\rightarrow \Omega$. In addition,  $\widetilde{\tens{T}}$ refers to the response function at a material point $\widetilde{\vect{X}}=\vect{\xi}(\vect{X})$ (see Fig.~\ref{fig:ref_conf}). By following the same procedure as in Section~\ref{sec:ass_2}, we find
\[
    \widetilde{\tens{T}}(\widetilde{\tens{F}},\,\widetilde{\tens{\Sigma}})=\tens{T}_0(\widetilde{\tens{F}}\tens{P}(\widetilde{\tens{\Sigma}})) = \widehat{\tens{T}}(\widetilde{\tens{F}},\,\widetilde{\tens{\Sigma}}),
\]
so that $\widetilde{\tens{T}}=\widehat{\tens{T}}$, namely $\widetilde{\tens{T}}$ \emph{is independent of the specific choice of the reference configuration}.  By further observing that $\widetilde{\tens{\Sigma}} = \widehat{\tens{T}}(\tens{H},\,\tens{\Sigma})$, we can write
\[
    \widehat{\tens{T}}(\tens{F},\,\tens{\Sigma}) = \widehat{\tens{T}}(\widetilde{\tens{F}},\,\widetilde{\tens{\Sigma}})=\widehat{\tens{T}}(\widetilde{\tens{F}},\,\widehat{\tens{T}}(\tens{H},\,\tens{\Sigma})),
\]
and, by using $\tens{F} = \widetilde{\tens{F}}\tens{H}$ together with the arbitrariness of
$\widetilde{\Omega}_0$, we finally get
\[
    \widehat{\tens{T}}(\widetilde{\tens{F}}\tens{H},\,\tens{\Sigma})=\widehat{\tens{T}}(\widetilde{\tens{F}},\,\widehat{\tens{T}}(\tens{H},\,\tens{\Sigma}))\quad\forall\widetilde{\tens{F}},\,\tens{H}\in\lin^+,\,\forall\tens{\Sigma}\in\sym,
\]
which is none other than ISRI.

\section{Structure of the energy density for isotropic initially stressed hyperelastic materials}
\label{sec:structure}

The goal of this section is deriving the general structure of the energy density for isotropic initially stressed hyperelastic materials. To do so, we first recall the main properties of isotropic materials.
A material is said to be \textit{isotropic} if the material symmetry group of the archetype is the whole group of rotations, namely
\[
    \tens{T}_0(\Fe\tens{R}) = \tens{T}_0(\Fe),\qquad\forall\Fe\in\lin^+,\,\tens{R}\in\orth^+.
\]
By referring to Fig \ref{fig:ref_conf}, the energy density of an isotropic, hyperelastic, material with respect to its relaxed state is usually given in terms of the left Cauchy-Green tensor $\tens{B}_\text{e}=\Fe\Fe^T $ or, equivalently, in terms of the right Cauchy-Green tensor $\tens{C}_\text{e} = \Fe^T \Fe$, as follows
\[
    f(\tens{B}_\text{e}) = f(\tens{C}_\text{e}) \coloneqq\psi_0(\Fe) \, .
\]
More specifically, thanks to the representation theorem of isotropic functions (see \cite{Gurtin_2010}, for instance), we can write
\begin{equation}
    \label{eq:energyFe}
    \psi_0(\Fe) = g(\tr(\Be),\,II(\Be),\,\det(\Be)) = g(\tr(\Ce),\,II(\Ce),\,\det(\Ce)) ,
\end{equation}
where $(\tr(\Be),\,II(\Be),\,\det(\Be))$ is the set of scalar invariants of $\Be$, collectively denoted with $\mathcal{I}_\tens{\Be}$, while $(\tr(\tens{C}_\text{e}),\,II(\tens{C}_\text{e}),\,\det(\tens{C}_\text{e}))$ is the set of scalar invariants of $\tens{C}_\text{e}$, collectively denoted with $\mathcal{I}_\tens{\tens{C}_\text{e}}$. We recall that it holds $\mathcal{I}_\tens{\Be}= \mathcal{I}_\tens{\tens{C}_\text{e}}$, being the expression of the second invariant
\[
    II(\Be)=\frac{\tr(\Be)^2-\tr(\Be^2)}{2} \, .
\]

Since our objective is to write \eqref{eq:energyFe} with respect to an initially stressed reference configuration, say $\Omega_0$ in Fig. \ref{fig:ref_conf}, we shall express its dependence upon $\tens{B}=\tens{F} \tens{F}^T$ (or $\tens{C}=\tens{F}^T \tens{F}$) and $\tens{\Sigma}$. To proceed in this direction, we  first notice that, since the response function $\tens{T}_\mathcal{F}$ is invertible by assumption, then $\tens{T}_\mathcal{F}$ is also semi-invertible \cite{Thiel_2019}. Consequently, there exists some function $\beta_j$, with $j\in\{0,\,1,\,2\}$ such that
\begin{equation}
    \label{eq:Be_inv}
    \Be = \beta_0(\mathcal{I}_\tens{T})\tens{I}+\beta_1(\mathcal{I}_\tens{T})\tens{T}+\beta_2(\mathcal{I}_\tens{T})\tens{T}^2,
\end{equation}
where $\mathcal{I}_\tens{T}$ is the set of invariants of $\tens{T}$.
In particular, by taking $\tens{F}=\tens{I}$ in \eqref{eq:Be_inv} and by noticing that $\Be =\F\tens{P}\tens{P}^T\F^T$, we get
\begin{equation}
    \label{eq:Be}
    \tens{P}\tens{P}^T = \beta_0(\mathcal{I}_\tens{\tens{\Sigma}})\tens{I}+\beta_1(\mathcal{I}_\tens{\Sigma})\tens{\Sigma}+\beta_2(\mathcal{I}_\tens{\Sigma})\tens{\Sigma}^2,
\end{equation}
where $\mathcal{I}_\tens{\Sigma}=(K_1,\,K_2,\,K_3)$ is the set of invariants of $\tens{\Sigma}$, i.e.
\[
    K_1 = \tr\tens{\Sigma},\quad K_2 = II(\tens{\Sigma}),\quad K_3=\det\tens{\Sigma}.
\]

By using \eqref{eq:Be_inv} and \eqref{eq:Be}, after some manipulations we have
\begin{gather}
    \label{eq:schifo1}
    \tr(\Be) = \beta_0 I_1 + \beta_1 I_4 + \beta_2 I_6 \, ,\\
    \label{eq:schifo2}
    \begin{gathered}
        \tr(\Be^2) =  \beta_0^2 (I_1^2-2I_2)+\beta_1^2 \tr(\tens{\Sigma}\C\tens{\Sigma}\C) + \beta_2^2 \tr(\tens{\Sigma}^2\C\tens{\Sigma}^2\C) \\
        +2\beta_0\beta_1 I_5 + 2\beta_0\beta_2 I_7 +2 \beta_1\beta_2\tr(\tens{\Sigma}\C\tens{\Sigma}^2\C) \, ,
    \end{gathered}\\
    \label{eq:schifo3}
    \det({\Be}) = I_3 \, \det(\tens{P} \tens{P}^T) \, ,
\end{gather}
where scalars $I_1, I_2, \dots, I_7$ are
\begin{equation}
    \label{eq:invariants}
    \begin{gathered}
        I_1 = \tr(\B)=\tr(\C) \, ,\qquad I_2 = II(\B)=II(\C) \, ,\qquad I_3 =\det(\B) \, ,\\
        I_4 = \tr(\tens{\Sigma}\C) \, ,\quad I_5 =\tr(\tens{\Sigma}\C^2) \, ,\quad I_6 = \tr(\tens{\Sigma}^2\C) \, ,\quad I_7 = \tr(\tens{\Sigma}^2\C^2) \, .
    \end{gathered}
\end{equation}

By noticing that $\tr(\tens{\Sigma}\C\tens{\Sigma}\C)$, $\tr(\tens{\Sigma}^2\C\tens{\Sigma}^2\C)$, and $\tr(\tens{\Sigma}\C\tens{\Sigma}^2\C)$ appearing in \eqref{eq:schifo2} can be obtained in terms of the invariants listed in  \eqref{eq:invariants} using the Cayley-Hamilton theorem (see \cite{Shams_2011}), we can write
\[
    \boxed{
        \tr({\Be}) = h_1 (I_1,\,I_4,\,I_6,\,\mathcal{I}_\tens{\Sigma})} \qquad \text{and} \qquad
    \boxed{II(\tens{B}_e) = h_2(I_1,\,I_2, \dots, I_7,\,\mathcal{I}_\tens{\Sigma})} \, .
\]
Moreover, since $\det(\tens{P} \tens{P}^T)$ can be expressed in terms of $\mathcal{I}_\tens{\Sigma}$, as reported in Appendix \ref{sec:appendix} for the sake of brevity, we have
\[
    \boxed{
        \det({\Be}) =I_3 h_3 (\mathcal{I}_\tens{\Sigma}) \, .  }
\]
Therefore, we have shown that the strain energy density of initially stressed hyperelastic materials simply depends on a set of independent invariants $I_1, I_2, \dots, I_7$ and $\mathcal{I}_\tens{\Sigma}$.
To sum up, we are able to write the general expression of $\widehat{\psi}(\tens{F},\,\tens{\Sigma})$ as
\[
    \widehat{\psi}(\tens{F},\,\tens{\Sigma}) = g(h_1(I_1,\,I_4,\,I_6,\,\mathcal{I}_\tens{\Sigma}),\,h_2(I_1,\,I_2, \dots, I_7,\,\mathcal{I}_\tens{\Sigma}),\,I_3 h_3 (\mathcal{I}_\tens{\Sigma})).
\]
We remark that $h_1$,\,$h_2$, and $h_3$ are not arbitrarily chosen but arise from \eqref{eq:schifo1}-\eqref{eq:schifo3} and the computations shown in Appendix~\ref{sec:appendix}.

In the following, we specialise such a constitutive description to the case of incompressible materials.

\section{The incompressible limit}
\label{sec:incompressible}

\subsection{Deviatoric--volumetric splitting of the stress tensors}

A material is said to be \textit{incompressible} if it cannot undergo deformations which result in local volume changes. Accordingly, the following condition holds
\begin{equation}
    \label{eq:incompressibility}
    J=\det\tens{F}=1 \, ,
\end{equation}
usually known as the \emph{incompressibility constraint}. We also introduce the deviatoric and volumetric part of the of the Cauchy stress tensor $\tens{T}$ as
\begin{equation} \label{eq:dev/vol}
    \tens{T}_\text{v} =\frac{\tr \tens{T}}{3} \, \tens{I} \, \qquad \text{and} \qquad \tens{T}_\text{d}=\tens{T}-\tens{T}_\text{v} \, ,
\end{equation}
where the quantity
\begin{equation}
    \label{eq:pressure}
    p=-1/3 \tr\tens{T} \, ,
\end{equation}
can be identified as a \emph{hydrostatic pressure}. While $\tens{T}_\tens{d}$ shall be constitutively provided, the pressure $p$ is a reactive term that can be regarded as a Lagrange multiplier enforcing the incompressibility constraint \eqref{eq:incompressibility}, in a sense that will be clarified later.

In the incompressible limit, the pressure expends no power during motion \cite{Gurtin_2010}. To see this, we introduce the spatial velocity $\vect{v}$ and its gradient $\tens{L}=\grad\vect{v}$, where $\grad$ and $\diver$ are the spatial gradient and spatial divergence operator, respectively. By denoting with $\tens{D}$ the symmetric part of $\tens{L}$, and by recalling that the expended power density can be computed as $(\tens{T}_\text{v}+\tens{T}_\text{d}):\tens{D}$, we obtain, for the volumetric part,
\[
    \tens{T}_\text{v}:\tens{D} = - p \tens{I}:\tens{D} = - p \diver \vect{v}=0 \, ,
\]
where we have used the fact that the incompressibility constraint \eqref{eq:incompressibility} implies $\diver\vect{v}=0$ \cite{Gurtin_2010}. Consequently, the pressure does not produce any work during motion and so cannot contribute to the elastic strain energy of an incompressible body.

Focusing now on initially stressed incompressible materials, we can also apply the volumetric-deviatoric splitting \eqref{eq:dev/vol} to the initial stress $\tens{\Sigma}$, the latter being  the stress tensor in the undeformed reference configuration, i.e. $\tens{\Sigma}=\left.\tens{T}\right|_{\tens{F}=\tens{I}}$. In this way, we can identify with $\Sigmad$ and $\Sigmav$ the deviatoric and volumetric part of $\tens{\Sigma}$, respectively as
\begin{equation}
    \label{eq:sigma_dv}
    \Sigmav = \frac{1}{3}(\tr\tens{\Sigma})\tens{I} \qquad \text{and} \qquad \Sigmad =\tens{\Sigma}-\Sigmav \, .
\end{equation}

By specializing the reasoning on $\tens{T}$ to $\tens{\Sigma}$, if we assume the initial stress be generated by elastic distortion, its volumetric part $\Sigmav$ does not contribute to the elastic energy. Accordingly, we take the strain energy density with respect to the initially stressed configuration to be a function of the deviatoric part of $\tens{\Sigma}$ and of the isochoric part of the deformation gradient $\oF = J^{-1/3}\tens{F}$, that is

\begin{equation}
    \label{eq:psibar}
    \boxed{
        \psi = \overline{\psi}(\oF,\,\Sigmad)}\, .
\end{equation}

One may observe that taking $\overline{\psi}$ as a function of $J^{-1/3}\tens{F}$ may seem useless since $J=1$. However, such a choice has two motivations.
\begin{enumerate}
    \item The first one is of theoretical nature: if we consider the Piola-Kirchhoff stress tensor $\tens{S}$, its deviatoric and volumetric parts read, respectively as
          \[
              \tens{S}_\text{d} = J \tens{T}_\text{d}\tens{F}^{-T}\qquad \text{and} \qquad \tens{S}_\text{v} = J \tens{T}_\text{v}\tens{F}^{-T},
          \]
          which, can also be expressed in terms of \eqref{eq:psibar} as follows  \cite{holzapfel2000nonlinear}
          \begin{equation} \label{eq:dev_vol_1stpiola}
              \tens{S}_\text{d} = \frac{\partial\overline{\psi}}{\partial\tens{F}} \qquad \text{and} \qquad {S_\text{d}}_{ij} = \frac{\partial\overline{\psi}}{F_{ij}}.
          \end{equation}
          In contrast, for a general $\widehat{\psi}(\tens{F},\,\tens{\Sigma})$
          \[
              \tens{S}_\text{d}\neq \frac{\partial\widehat{\psi}}{\partial\tens{F}}.
          \]
          Therefore, the introduction of $\oF$ allows us to separate the volumetric and the isochoric parts of the stress tensors in a natural way. This aspect will be useful in the following.

    \item The second reason is of numerical nature. When using the finite element method, it has been shown that the volumetric-deviatoric splitting of the stress tensor of the form \eqref{eq:sigma_dv} results in a more stable numerical formulation that reduces locking effects in incompressible materials. We will show that using \eqref{eq:psibar} improves the performances of the numerical algorithms for incompressible initially stressed materials as well.
\end{enumerate}

\subsection{Isotropic incompressible materials}
By following closely the reasoning of Section~\ref{sec:structure}, it is possible to write the energy density of an incompressible and isotropic initially stressed material as a function of a finite set of scalar invariants. Specifically, by starting with a general strain energy in the form \eqref{eq:psibar} along with the assumption of material isotropy, the resulting set of independent invariants yields

\textbf{\begin{equation}
        \label{eq:barredinvariants}
        \bar{I}_1,\,\bar{I}_2,\,\bar{I}_4,\,\bar{I}_5,\,\bar{I}_6,\,\bar{I}_7,\,\overline{K}_2,\,\overline{K}_3 \, ,
    \end{equation}}
where $\bar{I}_1, \, \bar{I}_2, \,  \dots , \bar{I}_7$ read
\begin{equation}
    \label{eq:Ijbar}
    \begin{gathered}
        \bar{I}_1=J^{-2/3} I_1, \qquad \bar{I}_2=J^{-4/3}I_2, \qquad\bar{I}_4 = J^{-2/3}\tr(\Sigmad \tens{C}),\\
        \bar{I}_5 = J^{-4/3}\tr(\Sigmad \tens{C}^2), \qquad\bar{I}_6 = J^{-2/3}\tr(\Sigmad^2 \tens{C}) ,\qquad\bar{I}_7 = J^{-4/3}\tr(\Sigmad^2 \tens{C}^2),
    \end{gathered}
\end{equation}
while $\overline{K}_2$ and $\overline{K}_3$ are the principal invariants of $\tens{\Sigmad}$, i.e.
\begin{equation}
    \label{eq:Kjbar}
    \overline{K}_2 = II(\Sigmad) = -\frac{\tr\Sigmad^2}{2} \qquad \text{and} \qquad \overline{K}_3 = \det\Sigmad \, .
\end{equation}

Compared with previous models (see \cite{Shams_2011,Gower_2015} for instance), such a formulation allows us to reduce the number of independent invariant by one since $\overline{K}_1 = \tr\Sigmad$ is zero by construction. To provide an example, we now derive the expression of the strain energy density for an incompressible initially stressed neo-Hookean material using the set of invariants \eqref{eq:barredinvariants}.

\subsection{Incompressible initially stressed neo-Hookean materials}

To start with, we consider the usual strain energy density of a neo-Hookean material with respect to the stress-free state depicted in Fig. \ref{fig:ref_conf}

\begin{equation}
    \label{eq:NHprestress}
    \psi_0(\Fe) =\frac{\mu}{2}(J_\text{e}^{-2/3}\tr\tens{B}_\text{e}-3) \, ,
\end{equation}
where $J_\text{e}=\det\tens{F}_\text{e}$. The resulting Piola-Kirchhoff stress is
\[
    \tens{S} = \frac{\partial\psi_0}{\partial\Fe}\tens{P}^{T} - p \tens{F}^{-T} = \mu J_\text{e}^{-2/3}\left(\Fe-\frac{1}{3}(\tr\tens{B}_\text{e})\Fe^{-T}\right)\tens{P}^{T} - p \tens{F}^{-T} \, .
\]

By using \eqref{eq:Cauchy} along with the condition $J_\text{e}=1$, the Cauchy stress tensor results
\[
    \tens{T} = \mu\left(\Be-\frac{1}{3}(\tr\Be)\tens{I}\right) - p \tens{I} \ ,
\]
so that its deviatoric and the volumetric parts yield
\[
    \tens{T}_d = \mu\left(\Be-\frac{1}{3}(\tr\tens{B}_\text{e})\tens{I}\right), \qquad\tens{T}_v = - p \tens{I}.
\]

To get the expression of $\tens{\Sigma}$ in the configuration $\Omega_0$, we enforce the compatibility condition \eqref{eq:ISCC}, thus getting
\begin{equation}
    \label{eq:comp_nh}
    \tens{\Sigma} = \mu\left(\tens{A}-\frac{1}{3}(\tr\tens{A})\tens{I}\right) - p_\tens{\Sigma}\tens{I},
\end{equation}
where $\tens{A} = \tens{P}\tens{P}^{T}$, while $p_\tens{\Sigma}$ is the initial hydrostatic pressure. By splitting the initial stress into its deviatoric and volumetric parts, we get
\begin{align}
    \label{eq:comp_sv}
     & \Sigmad = \mu\left(\tens{A}-\frac{1}{3}(\tr\tens{A})\tens{I}\right), \\
     & \Sigmav = - p_\tens{\Sigma}\tens{I}.
\end{align}

To derive the strain energy density in the form \eqref{eq:psibar}, we shall invert \eqref{eq:comp_sv} and express $\tens{A}$ as a function of the initial stress. Since $\tens{A}$ is symmetric, we can project \eqref{eq:comp_sv} along its principal directions, which coincide with the principal directions of $\Sigmad$. By denoting with $\sigma_j$ ($j=1,\,2,\,3$) the eigenvalues of $\Sigmad$, and with $\alpha_j$ ($j=1,\,2,\,3$) the eigenvalues of $\tens{A}$, \eqref{eq:comp_sv} rewrites as
\begin{equation}
    \label{eq:sigmas}
    \sistema{
        &\sigma_1 = \mu \left(\alpha_1-\frac{\alpha_1 +\alpha_2 +\alpha_3}{3}\right),\\
        &\sigma_2 = \mu \left(\alpha_2-\frac{\alpha_1 +\alpha_2 +\alpha_3}{3}\right),\\
        &\sigma_3 = \mu \left(\alpha_3-\frac{\alpha_1 +\alpha_2 +\alpha_3}{3}\right).\\
    }
\end{equation}
Since $\tens{A}$ is positive definite and $\det\tens{A} = (\det\tens{P})^{2}=1$, we have
$\alpha_j>0$ and
\begin{equation}
    \label{eq:incalpha}
    \alpha_1\alpha_2\alpha_3=1.
\end{equation}
We can subtract the second and the third equation of \eqref{eq:sigmas} from the first one, so that
\begin{equation}
    \label{eq:diff_sigma}
    \sistema{
        &\sigma_1 -\sigma_2 = \mu \left(\alpha_1-\alpha_2\right),\\
        &\sigma_1-\sigma_3 = \mu \left(\alpha_1-\alpha_3\right).\\
    }
\end{equation}
We can now use the following Proposition.

\begin{proposition}
    \label{prop:1}
    For all $t_{12},\,t_{13}\in\mathbb{R}$, the algebraic system of equations
    \begin{equation}
        \label{eq:sistemino}
        \sistema{
        &y_1-y_2=t_{12} \, ,\\
        &y_1-y_3=t_{13} \, ,\\
        &y_1y_2y_3=1  \, ,
        }
    \end{equation}
    has one and only one solution such that $y_j>0$ for $j=1,\,2,\,3$. In particular, $y_1$ is the maximal real root of the polynomial
    \[
        f(y_1) =  y_1(y_1-t_{12})(y_1-t_{13})-1 \, .
    \]
\end{proposition}
\begin{proof}
    Since all the unknowns $y_j$ must be positive, from the first two equations \eqref{eq:sistemino} we get the following restrictions
    \begin{equation}
        \label{eq:prop_1}
        \sistema{
        y_2=y_1-t_{12}>0,\\
        y_3=y_1-t_{13}>0.
        }
    \end{equation}
    In particular, $y_1>y_m\coloneqq\max\{0,\,t_{12},\,t_{13}\}$. By substituting \eqref{eq:prop_1} into the third equation of \eqref{eq:sistemino}, we get
    \[
        f(y_1)= y_1(y_1-t_{12})(y_1-t_{13})-1=0.
    \]
    We now observe that $f(y_m)=-1$ and $f'(y_1)=(y_1-t_{12})(y_1-t_{13})+y_1(y_1-t_{13})+y_1(y_1-t_{12})>0$ for $y_1>y_m$, therefore there exists one and only one positive solution of \eqref{eq:sistemino}.
\end{proof}

By Proposition~\ref{prop:1}, the system of equations \eqref{eq:incalpha}-\eqref{eq:diff_sigma} admits one and only one solution for $\alpha_j>0$, with $j=1,\,2,\,3$. In particular, such a solution corresponds to the maximal real root of the cubic polynomial
\begin{equation}
    \label{eq:falpha}
    f(\alpha_1)=\alpha_1\left(\alpha_1+\frac{\sigma_2-\sigma_1}{\mu}\right)\left(\alpha_1+\frac{\sigma_3-\sigma_1}{\mu}\right)-1.
\end{equation}
Furthermore, by using \eqref{eq:diff_sigma}, it is immediate to see that
\begin{equation}
    \label{eq:trA}
    \tr \tens{A} = 3 \alpha_1+\frac{\sigma_2+\sigma_3-2 \sigma_1}{\mu},
\end{equation}
so that, by substitution in \eqref{eq:falpha}, we get that $\tr\tens{A}$ can be obtained as the maximal real root of
\begin{equation}
    \label{eq:depressed_cubic}
    \xi^3+ \overline{K}_2 \xi+ \overline{K}_3-\mu^3 =0,
\end{equation}
where $\xi = \mu(\tr\tens{A})/3$. Equation \eqref{eq:depressed_cubic} is of the form $\xi^3 + \gamma \xi + \zeta = 0$, which is a depressed cubic equation. To calculate its roots, we introduce the discriminant $\Delta$, defined as
\[
    \Delta = \frac{\zeta^2}{4}+\frac{\gamma^3}{27}.
\]
Based on the actual value of $\Delta$ we have two cases.
\begin{itemize}
    \item \textbf{Case I}: $\Delta>0$. We have one real root of \eqref{eq:depressed_cubic}, given by the del Ferro-Tartaglia formula
          \[
              \xi = \sqrt[3]{-\frac{\zeta}{2}+\sqrt{\Delta}}+\sqrt[3]{-\frac{\zeta}{2}-\sqrt{\Delta}}
          \]
    \item \textbf{Case II}: $\Delta \leq 0$. Let $\phi$ be the argument of the complex number $-\zeta/2+i\sqrt{-\Delta}$. Then the three real roots of \eqref{eq:depressed_cubic} are given by
          \[
              \xi_j = 2 \sqrt{-\frac{\gamma}{3}}\cos \frac{\phi+2j\pi}{3},\qquad\text{with }j=0,\,1,\,2,
          \]
where the only acceptable solution $\xi$ is given by $\xi=\max\{\xi_1,\,\xi_2,\,\xi_3\}$.
\end{itemize}

To conclude the derivation, we observe that
\begin{equation}
    \label{eq:trfe}
    \tr\Be = \tr(\tens{A}\tens{C}).
\end{equation}
We then multiply equation \eqref{eq:comp_sv} by $\tens{C}$
\begin{equation}
    \label{eq:sc}
    \mu\tens{A}\tens{C} = \Sigmad\tens{C} + \xi\tens{C},
\end{equation}
whose trace is

\begin{equation} \label{eq:trBe}
    \mu \tr\Be=\mu\tr(\tens{A}\tens{C}) = \tr(\Sigmad\tens{C}) + \xi(\tr\tens{C}).
\end{equation}
Therefore, by substitution of \eqref{eq:trBe} into \eqref{eq:NHprestress} we can calculate the strain energy density in terms of the invariants \eqref{eq:barredinvariants}, that is

\begin{equation}
    \label{eq:strainenergy}
    \psi = \overline{\psi}(\oF,\,\Sigmad) = \frac{1}{2}( \xi \overline{I}_1+\overline{I}_4-3\mu),
\end{equation}
and the associated Piola-Kirchhoff stress tensor
\[
    \tens{S} = J^{-2/3}\left( \xi \tens{F}+\tens{F}\Sigmad\right)-\left(\frac{1}{3}\left(\xi\overline{I}_1+\overline{I}_4\right)+J^{-2/3}p\right)\tens{F}^{-T}.
\]

The strain energy density given in \eqref{eq:strainenergy} closely resembles the structure of the strain energy initially proposed by Gower and co-authors (see Eq. (3.11) in \cite{Gower_2015}). However, there are notable differences. In \cite{Gower_2015}, the term $\mathring{p}$, representing the pressure field in the reference configuration, is used instead of $\xi$. The equation for $\xi$ is somewhat simpler, as it is a depressed cubic equation compared to the equation for $\mathring{p}$. Additionally, Proposition~\ref{prop:1} enables us to rigorously identify the unique acceptable solution among the three roots of $\xi$.\\
Furthermore, the use of barred invariants given in \eqref{eq:barredinvariants} facilitates a more efficient finite element formulation, as demonstrated in the following sections. Before presenting the numerical results, we specialize the proposed theoretical framework to plane strain deformations.
\subsection{Plane strain}
\label{sec:plane_strain}
We now assume plane strain deformations both during the generation of the initial stress and for the subsequent elastic deformation.
Thus, let $(\vect{e}_1,\,\vect{e}_2,\,\vect{e}_3)$ be an orthonormal right-handed vector basis, we assume that
\begin{equation}
    \label{eq:plane_strain}
    F_{13}=F_{23}=F_{31}=F_{32}=P_{13}=P_{23}=P_{31}=P_{32}=0,\quad F_{33}=P_{33}=1,
\end{equation}
where $M_{ij}=\vect{e}_i\cdot\tens{M}\vect{e}_j$ and $\tens{M}$ is a second order tensor. Under this assumption, from \eqref{eq:comp_nh} it results that $\Sigma_{13}= \Sigma_{23}=0$. Moreover, we get that one of the eigenvalues of $\tens{A}$, say $\alpha_3$, is equal to one.
Therefore, \eqref{eq:incalpha} reduces to $\alpha_1\alpha_2=1$, which by \eqref{eq:diff_sigma} becomes
\begin{equation}
    \label{eq:incalphaplane}
    \alpha_1\left(\alpha_1+\frac{\sigma_2-\sigma_1}{\mu}\right)=1.
\end{equation}
Without loss of generality, we can take $\sigma_1>\sigma_2$, so that the only admissible solution of \eqref{eq:incalphaplane} is given by
\[
    \alpha_1=\frac{\sqrt{4 \mu ^2+(\sigma_1-\sigma_2)^2}+\sigma_1-\sigma_2}{2 \mu }.
\]
Thus, we get
\begin{equation}
    \label{eq:xi_2D}
    \xi = \frac{\mu}{3}\tr\tens{A} = \frac{\sqrt{4 \mu ^2+(\sigma_1-\sigma_2)^2}+\mu }{3 }.
\end{equation}
The principal deviatoric stresses rewrite from \eqref{eq:sigmas} as
\begin{equation} \label{eq:sigmas_plane}
    \sistema{
        &\sigma_1 = \frac{1}{6} \left(3 (\sigma_1 - \sigma_2 ) - 2 \mu + \sqrt{4 \mu ^2+(\sigma_1-\sigma_2)^2}\right),\\
        &\sigma_2 =\frac{1}{6} \left( - 3 (\sigma_1 - \sigma_2 ) - 2 \mu + \sqrt{4 \mu ^2+(\sigma_1-\sigma_2)^2}\right) ,\\
        &\sigma_3 = \frac{1}{3} \left(2 \mu -\sqrt{4 \mu ^2+(\sigma_1-\sigma_2)^2}\right).
    }
\end{equation}
Note that both $\xi$ and $\Sigmad$ can be expressed as a function of $\sigma_1 - \sigma_2$, as can be seen from \eqref{eq:xi_2D} and \eqref{eq:sigmas_plane}. Such a quantity can be computed from the invariants of the initial stress by introducing $\tens{\Sigma}^\parallel$ as the restriction of $\tens{\Sigma}$ to the plane spanned by $(\vect{e}_1,\,\vect{e}_2)$.
We have
\begin{equation}
    \label{eq:s1s2}
    \sigma_1-\sigma_2 = \sqrt{ \left(\tr\tens{\Sigma}^\parallel\right)^2 - 4 \det\tens{\Sigma}^\parallel}.
\end{equation}

Therefore, under the plane strain conditions detailed in \eqref{eq:plane_strain}, the modified free energy \eqref{eq:strainenergy} is univocally identified given $\tens{\Sigma}^\parallel$. Operationally, after fixing $\tens{\Sigma}^\parallel$, $\xi$ and the principal deviatoric stresses $\sigma_1$, $\sigma_2$ and $\sigma_3$ are computed by means of \eqref{eq:xi_2D}, \eqref{eq:sigmas_plane} and \eqref{eq:s1s2}. Furthermore, the principal stresses of $\tens{\Sigma}^\parallel$ yield
\[
    \Sigma_{1,2} = \frac{1}{2} \left( \tr\tens{\Sigma}^\parallel \pm \sqrt{ \left(\tr\tens{\Sigma}^\parallel\right)^2 - 4 \det\tens{\Sigma}^\parallel} \right) \, ,
\]so that the value of the resulting initial pressure $p_{\tens{\Sigma}}$ can be calculated as
\[
    p_{\tens{\Sigma}} = \sigma_1 - \Sigma_1 \quad \text{or, equivalently,} \qquad p_{\tens{\Sigma}} = \sigma_2 - \Sigma_2.
\]
The non-vanishing components of $\Sigmad$ thus yield
\[
    \Sigma^d_{11} = \Sigma_{11} + p_{\tens{\Sigma}} \, , \, \Sigma^d_{22} = \Sigma_{22} + p_{\tens{\Sigma}} \, , \, \Sigma^d_{12} = \Sigma_{12}  \, , \, \Sigma^d_{33} = \Sigma_3.
\]
As a test, we propose the analysis of the bending of a rectangle subject to a plane initial stress.

\section{Bending of an initially stressed rectangular block}
\label{sec:bending}
In this Section, we construct an analytical solution corresponding to a pure bending of a rectangular block. Such a solution will be used as a benchmark for the numerical implementation of the formulation proposed in Section \ref{sec:incompressible} to show the advantages of the volumetric-deviatoric splitting.

\subsection{Analytical solution}

As depicted in Fig. \ref{fig:geometryBC}, we consider as a reference configuration the set
\[
    \Omega_0 = \left\{\vect{X}=(X,\,Y,\,Z)\in\mathbb{E}^3\;|\;-L/2\leq X\leq L/2,\,0\leq Y\leq H,\,0\leq Z\leq 1\right\},
\]
where $X,\,Y,\,Z$ are the Cartesian coordinates with respect to an orthonormal right-handed vector basis $(\vect{e}_X, \,\vect{e}_Y, \, \vect{e}_Z )$. The current configuration is instead described by means of a cylindrical system $(r,\,\theta,\,z)$ with $(\vect{e}_r,\,\vect{e}_\theta,\,\vect{e}_z)$ indicating the corresponding vector basis. We further assume plane strain deformations as discussed in Section~\ref{sec:plane_strain}, so that $z=Z$. Boundary conditions are specified as
\[
    \left\{
    \begin{aligned}
         & \theta = 0                              &  & \text{for }Y=0,     \\
         & \theta = \alpha                         &  & \text{for }Y=H,     \\
         & \tens{S}\vect{e}_Y = \vect{0}           &  & \text{for }X=-L/2,\,L/2, \\
         & \vect{e}_{r}\cdot\tens{S}\vect{e}_X = 0 &  & \text{for }Y=0,\,H.
    \end{aligned}
    \right.
\]

Following \cite{ogden1997non}, we then look for a homogeneous solution having the following form
\[
    r = f(X),\quad \theta=g(Y) = \frac{\alpha Y}{H},
\]
which represents a homogeneous bending of $\Omega_0$.
The deformation gradient thus reads
\begin{gather}
    \label{eq:Fbend}
    \tens{F}=f'(X)\vect{e}_r\otimes\vect{e}_X+f(X)g'(Y)\vect{e}_\theta \otimes\vect{e}_Y + \vect{e}_z\otimes\vect{e}_Z\\
    =f'(X)\vect{e}_r\otimes\vect{e}_X+\frac{\alpha f(X)}{H}\vect{e}_\theta\otimes\vect{e}_Y+\vect{e}_z\otimes\vect{e}_Z \, . \nonumber
\end{gather}
From the incompressibility constraint \eqref{eq:incompressibility}, we get $f'(X)f(X) = H / \alpha$, which gives
\[
    \begin{gathered}
        r  = f(X) = \sqrt {c_1 + \frac{2 L X}{\alpha}},\\  \qquad r_A = f (-L/2) = \sqrt{c_1 - \frac{L H}{\alpha}}, \qquad r_B = f (L/2) = \sqrt {c_1 + \frac{L H}{\alpha}},
    \end{gathered}
\]
where $c_1$ is a constant to be fixed by the boundary conditions once the material behaviour is specified.

In order to proceed, we need to prescribe some constitutive assumptions on the material response and the initial stress state.
We assume that the block is composed of an incompressible, initially stressed, neo-Hookean material, where the strain energy density is given by \eqref{eq:strainenergy}. As an example, we assume an axial initial stress of the form
\[
    \tens{\Sigma}^{\parallel} = \Sigma_{YY}(X)\vect{e}_Y\otimes\vect{e}_Y.
\]
We immediately have that $\Diver \tens{\Sigma} = \vect{0}$. From \eqref{eq:xi_2D} and \eqref{eq:s1s2} we get
\[
    \begin{aligned}
         & \xi=\frac{1}{3} \left(\mu +\sqrt{4 \mu ^2+\Sigma_{YY} ^2}\right) \, ,                                          \\
         & p_{\tens{\Sigma}} = - \frac{\Sigma_{YY} }{2} - \frac{\mu}{3} + \frac{1}{6} \sqrt{4 \mu^2 + \Sigma_{YY}^2} \, , \\
         & \sigma_3 = \mu - \xi \,.
    \end{aligned}
\]
Therefore, we get the following expression for the deviatoric part of the initial stress tensor $\Sigmad$
\[
    \Sigmad = p_\tens{\Sigma} \vect{e}_X\otimes\vect{e}_X + ( \Sigma_{YY} + p_\tens{\Sigma} ) \vect{e}_Y \otimes\vect{e}_Y + \sigma_3 \vect{e}_Z \otimes\vect{e}_Z \, .
\]
By \eqref{eq:Fbend}, the left Cauchy-Green tensor results
\[
    \tens{B} = \tens{F} \tens{F}^T = \frac{H^2}{r^2 \, \alpha^2} \vect{e}_r\otimes\vect{e}_r + \frac{r^2 \, \alpha^2 }{H^2} \vect{e}_\theta \otimes\vect{e}_\theta + \vect{e}_z\otimes\vect{e}_z,
\]
so the invariants $\bar{I}_1$ and $\bar{I}_4$ take the following expressions
\begin{equation}   \label{eq:r0_r1}
    \bar{I}_1 = \frac{H^2}{r^2 \, \alpha^2} + \frac{r^2 \, \alpha^2 }{H^2} + 1  \qquad \text{and} \qquad \bar{I}_4 = p_\tens{\Sigma}\frac{H^2}{r^2 \, \alpha^2}  + ( \Sigma_{YY} + p_\tens{\Sigma} ) \frac{r^2 \, \alpha^2 }{H^2} + \sigma_3.
\end{equation}

The balance of linear momentum in the current configuration reads $\diver\tens{T} = \vect{0}$, which in cylindrical coordinates and under the kinematic assumption of this section reduces to the scalar equations
\begin{align}
    \label{eq:eq:radial1}
     & \frac{\partial T_{rr}}{\partial r}+\frac{1}{r}(T_{rr}-T_{\theta\theta})=0,    \\
     & \frac{\partial T_{\theta \theta}}{\partial \theta} = 0, \label{eq:eq:radial2} \\
    \label{eq:eq:radial3}
     & \frac{\partial T_{zz}}{\partial z} = 0,
\end{align}
where the non-vanishing components of the Cauchy stress tensor yield
\begin{align}
     & T_{rr} = \left(p_\tens{\Sigma} + \xi  \right) \frac{H^2}{r^2 \, \alpha^2} -\frac{1}{3} \left( \bar{I}_4 + \xi \bar{I}_1 \right) - p , \label{eq:cauchy_rr}                                   \\
     & T_{\theta \theta } = \left(\Sigma_{YY} + p_\tens{\Sigma} + \xi \right) \frac{r^2 \, \alpha^2 }{H^2} -\frac{1}{3} \left( \bar{I}_4 + \xi \bar{I}_1 \right) - p , \label{eq:cauchy_thetatheta} \\
     & T_{zz } = \left( \sigma_3  + \xi  \right) -\frac{1}{3} \left( \bar{I}_4 + \xi \bar{I}_1 \right) - p  .
\end{align}
Notice that, owing to the kinematic assumptions and to the choice of the initial stress, Eqs. \eqref{eq:eq:radial2}-\eqref{eq:eq:radial3} are trivially satisfied. By integrating \eqref{eq:eq:radial1} and applying the boundary condition $T_{rr} (r_A) = 0$, we get
\[
    T_{rr} (r) = - \int_{r_A}^{r} \frac{1}{r} \left( T_{rr} - T_{\theta \theta} \right) \d r.
\]
A direct substitution of \eqref{eq:cauchy_rr}-\eqref{eq:cauchy_thetatheta} allows us to find the pressure $p$ as
\[
    p (r) = \left( p_\tens{\Sigma} (r) + \xi(r)  \right) \frac{H^2}{r^2 \, \alpha^2} -\frac{1}{3} \left( \bar{I}_4 (r) + \xi (r) \bar{I}_1 (r) \right) - T_{rr} (r).
\]

Finally, the constant $c_1$ can be computed by enforcing the boundary condition $T_{rr}(r_B)=0$.

\begin{figure}[t!]
    \centering
    \includegraphics[width=0.8\textwidth]{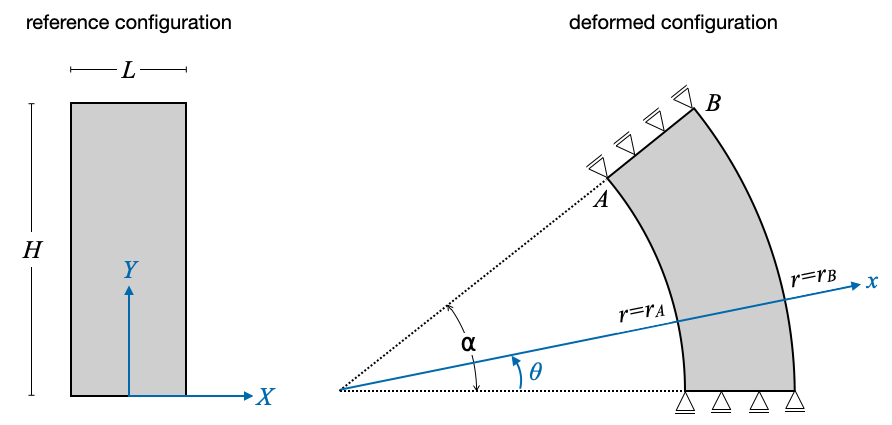}
    \caption{Schematic of the geometry and boundary conditions adopted for the bending of a rectangular block.}
    \label{fig:geometryBC}
\end{figure}

\begin{figure}[t!]
    \centering
    \includegraphics[width=1\textwidth]{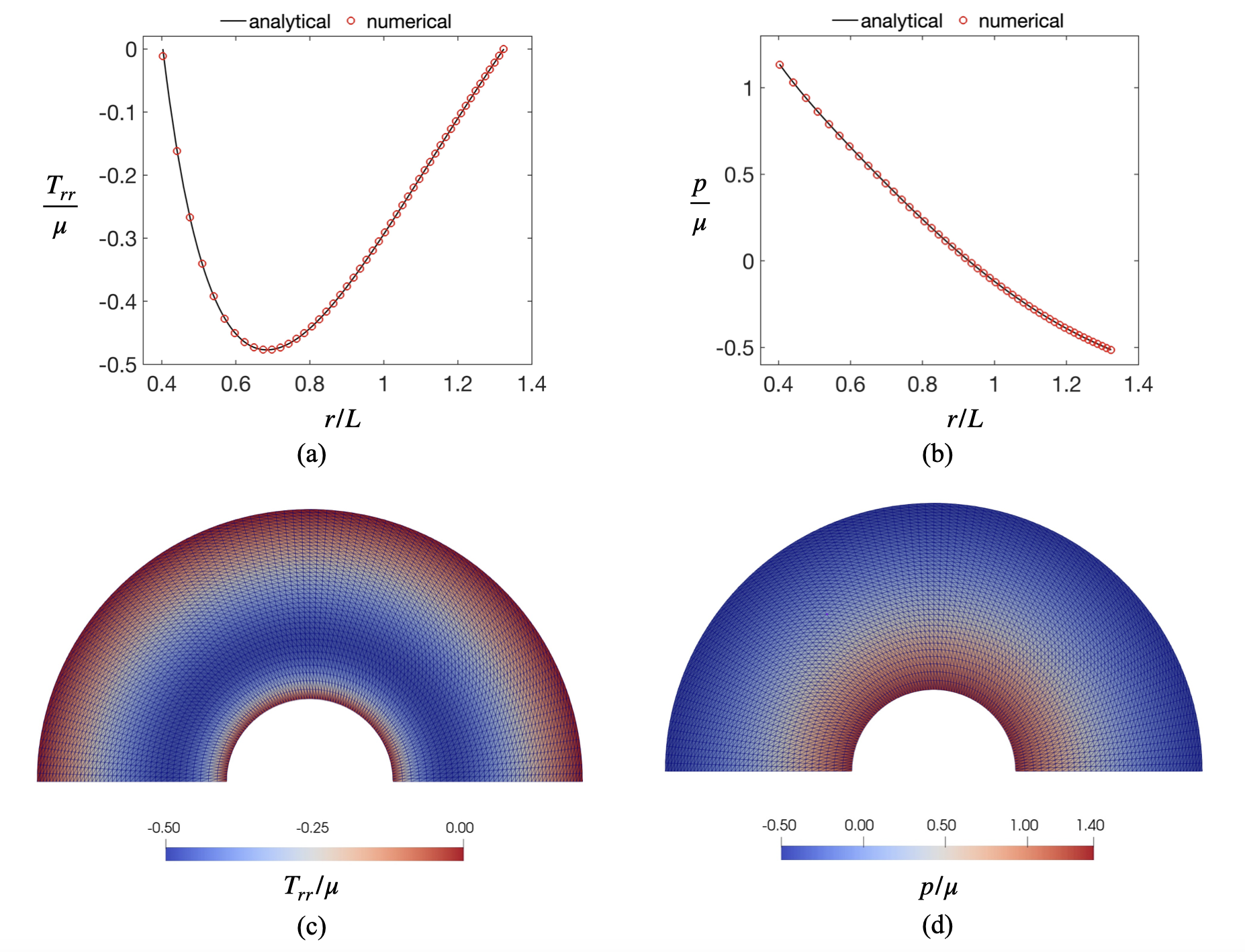}
    \caption{Bending of a rectangular block with initial stress: validation of the proposed numerical implementation using the finite element method. (a-b) Plot of the normalized radial stress $T_{rr}/\mu$ and hydrostatic pressure $p/\mu$ against the radial coordinate $r$ for both the analytical and numerical solutions. (c-d) Spatial distribution of the normalized radial stress and hydrostatic pressure in the deformed configuration for the adopted  finite element mesh.}
    \label{fig:validation}
\end{figure}

\begin{figure}[t!]
    \centering
    \includegraphics[width=1\textwidth]{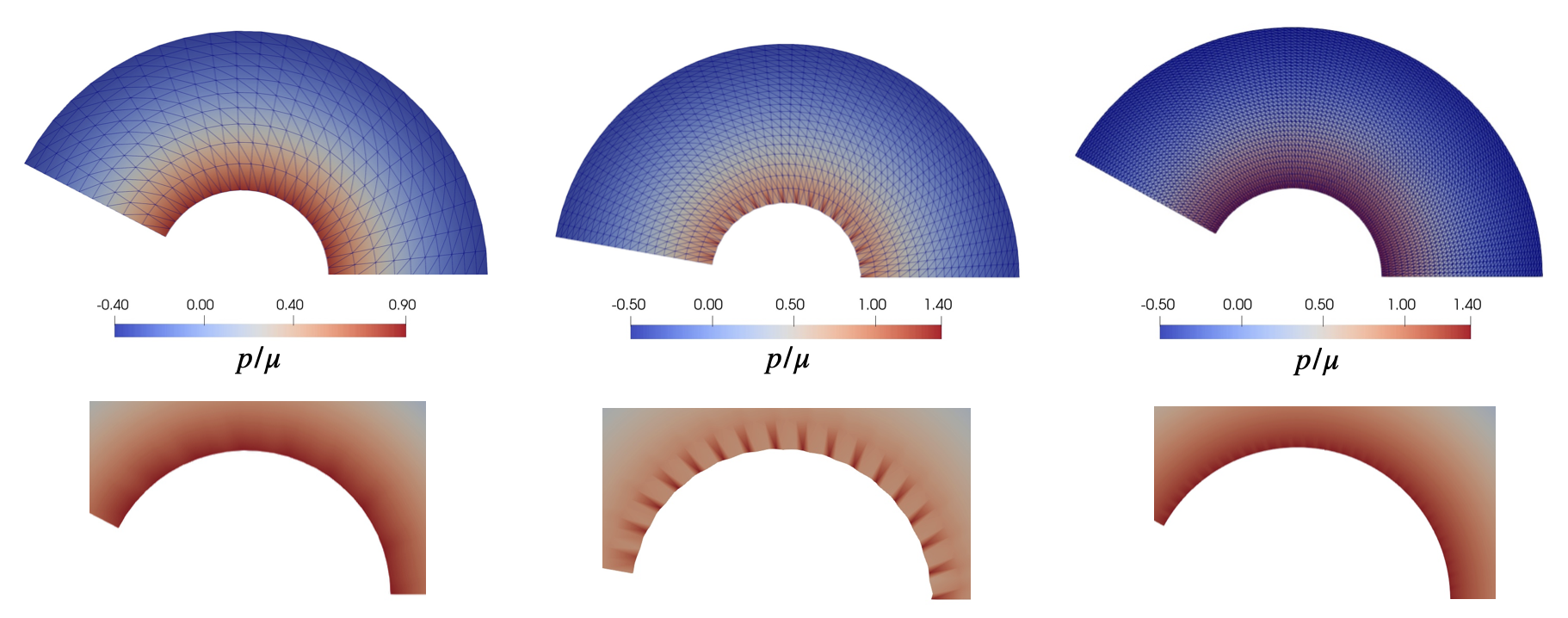}
    \caption{Bending of a rectangular block with initial stress: numerical results obtained using a standard model formulation without volumetric-deviatoric splitting. (Top) Simulated spatial distribution of the  normalized hydrostatic pressure for three different finite element meshes. The reported results refer to the last simulated load increment before the algorithm fails to converge. (Bottom) Magnification of the pressure distribution in the inner radius showing the typical mesh dependent fluctuations resulting from volumetric locking.}
    \label{fig:locking}
\end{figure}

\subsection{Finite element approximation}

In order to investigate the advantages of the volumetric-deviatoric splitting in a numerical setting, we compare the analytical solution with the one obtained by a finite element approximation of the governing equations. In particular, we implemented the model equations into the open-source computing platform \texttt{FEniCS} \cite{FenicsProject} for their numerical resolution through the finite element method. The geometry and boundary conditions of the problem are depicted in Fig. \ref{fig:geometryBC}. To handle the prescribed bending deformation, we impose the angle between the $X$-axis and side AB incrementally till its magnitude equals $\alpha$. Operationally, we impose such a constraint in a variational manner through the introduction of an \textit{ad-hoc} boundary Lagrange multiplier by using the \texttt{multiphenics} library \cite{multiphenics}. Therefore, the overall finite element formulation accounts for three independent variables, i.e. the displacement field $\boldsymbol{u}$, the hydrostatic pressure $p$, and the boundary Lagrange multiplier $\lambda$. We approximate $\boldsymbol{u}$ and $\lambda$ with second order elements and $p$ with first order elements.

In the numerical simulations, we select $L=2$, $H=5$, $\alpha=\pi$, $\mu=1$, and we choose the initial stress component
\[
    \Sigma_{YY} (X) = - \mu \frac{2 X}{L}.
\]
Figs. \ref{fig:validation}a-b show the resulting radial stress $T_{rr}$ and hydrostatic pressure $p$ against the radial coordinate $r$ for both the numerical and analytical solutions. We observe an excellent agreement between the two solutions, thus demonstrating the correctness of the proposed formulation. We also report the spatial distribution of $T_{rr}$ and $p$ in the deformed configuration in Figs. \ref{fig:validation}c-d for the adopted finite element mesh. As expected, the considered stress components distribute in space following a cylindrical symmetry.

The proposed formulation actually outperforms the conventional models for initially stressed neo-Hookean materials (see Eq. (3.11) in \cite{Gower_2015}) as shown in Fig. \ref{fig:locking}. Here, we gather the numerical outcomes for three different finite element meshes obtained using using the neo-Hookean model described by Eq. (3.11) in \cite{Gower_2015} without resorting to the volumetric-deviatoric splitting proposed in this paper. Such a formulation is unable to capture the imposed bending deformation since the algorithm fails to converge before reaching the final load increment. Lack of numerical convergence follows from spurious volumetric locking as visible in the simulated spatial distribution of the normalized hydrostatic pressure near the inner radius (see bottom of Fig. \ref{fig:locking}).

\section{Concluding Remarks}
In this study, we delved into critical aspects associated with the modelling of initially stressed materials, specifically focusing on a class of materials meeting two essential conditions:
\begin{enumerate}
    \item[A1.] We required invertibility of the response function $\tens{T}_0(\Fe)$ for the relaxed material, as elucidated in Section~\ref{sec:assumptions} (see Remark~\ref{rem:invert} for further discussion on this invertibility).
    \item[A2.] We explored materials wherein the initial stress arises from an elastic distortion.
\end{enumerate}

Our investigation highlighted the necessity for the invertibility of $\tens{T}_0(\Fe)$ to properly define $\widehat{\tens{T}}(\tens{F}, \, \tens{\Sigma})$. Moreover, we demonstrated that the commonly imposed \emph{initial stress reference independence} and \emph{initial stress compatibility condition} stem naturally from the above two assumptions. While (A1) is a prerequisite for defining $\widehat{\tens{T}}(\tens{F},\,\tens{\Sigma})$, (A2) may not hold in certain materials where stress formation is accompanied by a change in material properties, as discussed in \cite{Riccobelli_2019}.

Drawing inspiration from Hoger's work \cite{Hoger1993b}, we derived the general expression of the strain energy density $\widehat{\psi}(\tens{F},\,\tens{\Sigma})$ for materials satisfying (A1)-(A2). Notably, our analysis revealed that $\widehat{\psi}(\tens{F},\,\tens{\Sigma})$ may depend on $I_5$ and $I_7$ only if $\psi_0$ is a function of $II(\Be)$.

We then focused on incompressible media, demonstrating that the strain energy density can be expressed in terms of the deviatoric part of the initial stress tensor, $\Sigmad$, and the isochoric part of the deformation gradient, $\oF$. This volumetric-deviatoric splitting enabled a reduction of necessary invariants by one compared to existing models.

To illustrate the applicability of our approach, we applied it to an incompressible neo-Hookean material (see \eqref{eq:strainenergy}). In order to show the advantages of the volumetric-deviatoric splitting in a computational setting, we examined the bending of an initially stressed rectangular block. The numerical simulations demonstrated that the volumetric-deviatoric splitting of energy exhibits superior performance over the original formulation based on the energy proposed by \cite{Gower_2015}, eliminating locking phenomena and aligning perfectly with the analytical solution.

The outcomes of this research contribute to the modelling of initially stressed materials, including biological tissues and gels. The proposed volumetric-deviatoric splitting of the energy density holds significant potential in computational mechanics, offering a substantial enhancement to existing methodologies.

\appendix
\section{Expression of $\det \tens{P} \tens{P}^T$ in terms of $\mathcal{I}_{\tens{\Sigma}}$} \label{sec:appendix}

To start with, we conveniently indicate with $\tens{A}=\tens{P}\tens{P}^T$. By application of the Cayley-Hamilton theorem, we get
\[
    - \tens{A}^3 + \tr(\tens{A}) \, \tens{A}^2 - II(\tens{A}) + \det (\tens{A}) \tens{I} = \tens{0} \, ,
\]
and, by taking the trace of both sides along with exploiting the definition of $II(\tens{A}) $, it results

\begin{equation} \label{eq:detA}
    \det( \tens{A}) = \frac{1}{3} \left( \tr (\tens{A}^3) - \tr (\tens{A}) \, \tr (\tens{A}^2) + \tr (\tens{A}) \frac{\tr( \tens{A})^2 - \tr (\tens{A}^2)}{2}\right) \, .
\end{equation}

To express \eqref{eq:detA} in terms of $\mathcal{I}_{\tens{\Sigma}}$, we exploit \eqref{eq:Be}, from which it results, after some manipulations,
\begin{gather*}
    \tr(\tens{A}) = 3  \beta_0 + \beta_1 \tr(\tens{\Sigma}) + \beta_2 (\tr(\tens{\Sigma})^2 - 2 II(\tens{\Sigma}))  \, , \\
    \begin{gathered}
        \tr(\tens{A}^2) = 3 \beta_0^2 + (\beta_1^2+2 \beta_0 \beta_2) \left(\tr (\tens{\Sigma})^2 - 2 II (\tens{\Sigma}) \right) +\\
        + \beta_2^2 \tr (\tens{\Sigma}^4) + 2 \beta_0 \beta_1 \tr (\tens{\Sigma}) + 2 \beta_1 \beta_2 \tr(\tens{\Sigma}^3),
    \end{gathered}\\
    \tr(\tens{A}^3) = 3 \beta_0^3 + 3 \beta_0^2 \beta_1 \tr (\tens{\Sigma}) + 3 \beta_0 (\beta_1^3+\beta_2 ) \left(\tr(\tens{\Sigma})^2 - 2 II (\tens{\Sigma}) \right) + \beta_1^3 \tr (\tens{\Sigma}^3) \\
    +3 \beta_2 (\beta_0 \beta_2 + \beta_1^2) \tr(\tens{\Sigma}^4) + 3 \beta_1 \beta_2^2 \tr (\tens{\Sigma}^5) + \beta_2^3 \tr(\tens{\Sigma}^6) \, .
\end{gather*}
where $\tr(\tens{\Sigma}^3)$, $\tr(\tens{\Sigma}^4)$, $\tr(\tens{\Sigma}^5)$, and $\tr(\tens{\Sigma}^6)$ can be recast in therm of $\mathcal{I}_{\tens{\Sigma}}$ by repeatedly using the Cayley-Hamilton theorem for $\tens{\Sigma}$.

\section*{Acknowledgments}
DR gratefully acknowledge the support provided by the European Union – NextGenerationEU under the National Recovery and Resilience Plan (NRRP), Mission 4 Component 2 Investment 1.1 - Call PRIN 2022 No. 104 of February 2, 2022 of Italian Ministry of University and Research; Project 202249PF73 (subject area: PE - Physical Sciences and Engineering) ``Mathematical models for viscoelastic biological matter''. MM gratefully acknowledge the support provided by the European Commission through FSE REACT-EU funds, PON Ricerca e Innovazione.
The authors are members of the Gruppo Nazionale di Fisica Matematica – INdAM. The authors acknowledge Artur L. Gower for useful discussions.

\end{document}